\documentclass[conference]{IEEEtran}
\IEEEoverridecommandlockouts
\usepackage{cite}
\usepackage{amsmath,amssymb,amsfonts,amsthm, color}
\usepackage{graphicx}
\usepackage{textcomp}
\usepackage{xcolor}
\usepackage{enumitem}

\newtheorem*{theorem*}{Theorem}
\newtheorem{theorem}{Theorem}

\newtheorem{proposition}{Proposition} 
 
\newtheorem{lemma}{Lemma} 
\newtheorem{example}{Example}

\usepackage[utf8]{inputenc}
\usepackage{xcolor}

\usepackage[noend]{algpseudocode}
\usepackage{mathtools}
\usepackage{cuted}
\usepackage{caption}
\def\BibTeX{{\rm B\kern-.05em{\sc i\kern-.025em b}\kern-.08em
    T\kern-.1667em\lower.7ex\hbox{E}\kern-.125emX}}
\begin{document}

\title{Reconstruction and Error-Correction Codes for Polymer-Based Data Storage\\
{
}
\thanks{The work was funded by the DARPA Molecular Informatics program, the SemiSynBio program of the NSF and SRC, and the
NSF CIF grant 1618366.}
}

\author{\IEEEauthorblockN{Srilakshmi Pattabiraman}
\IEEEauthorblockA{\textit{ECE Department, UIUC} \\
Urbana, IL, USA \\
sp16@illinois.edu}
\and
\IEEEauthorblockN{Ryan Gabrys}
\IEEEauthorblockA{\textit{ECE Department, UCSD} \\
San Diego, CA, USA \\
ryan.gabrys@gmail.com}
\and
\IEEEauthorblockN{Olgica Milenkovic}
\IEEEauthorblockA{\textit{ECE Department, UIUC} \\
Urbana, IL, USA \\
milenkovic@illinois.edu}
}

\maketitle

\begin{abstract}
Motivated by polymer-based data-storage platforms that use chains of binary synthetic polymers as the recording media and read the content via tandem mass spectrometers, we propose a new family of codes that allows for unique string reconstruction and correction of one mass error. Our approach is based on introducing redundancy that scales logarithmically with the length of the string and allows for the string to be uniquely reconstructed based only on its erroneous substring composition multiset. 
The key idea behind our unique reconstruction approach is to interleave Catalan-type paths with arbitrary binary strings and ``reflect'' them so as to allow prefixes and suffixes of the same length to have different weights. For error correction, we add a constant number of bits that provides information about the weights of reflected pairs of bits and hence enable recovery from a single mass error. The asymptotic code rate of the scheme is one, and decoding is accomplished via a simplified version of the backtracking algorithm used for the Turnpike problem.
\end{abstract}


\vspace{-0.1in}
\section{Introduction}
Current digital storage systems are facing numerous obstacles in terms of scaling the storage density and allowing for in-memory based computations~\cite{zhirnov2016nucleic}. To offer storage densities at nanoscale, several molecular storage paradigms have recently been put forward in 
~\cite{al2017mass,goldman2013towards,grass2015robust,yazdi2015rewritable,yazdi2017portable}. One promising line of work with low storage cost and readout latency is the work in~\cite{al2017mass}, which proposed using synthetic polymers for storing user-defined information and reading the content via tandem mass spectrometry (MS/MS) techniques. More precisely, binary data is encoded using poly(phosphodiester)s, synthesized through automated phosphoamidite chemistry in such a way that the two bits $0$ and $1$ are represented by molecules of different masses that are stitched together into strings of fixed length. To read the encoded data, inter phosphate bonds are broken, and MS/MS readers are used to estimate the masses of the fragmented polymer and reconstruct the recorded string, as illustrated in Figure~\ref{fig:spectrometry}.
\begin{figure}[h]
\centering
\includegraphics[scale=0.3]{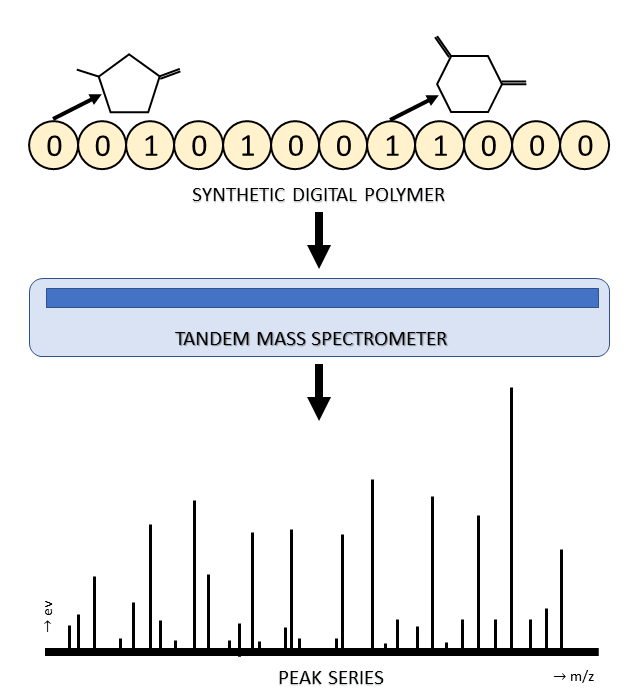} 
\caption{The scheme is adapted from~\cite{al2017mass}. The top figure depicts a binary string synthesized using phosphoamide chemistry. The bottom image is an illustration of \emph{peak series} or MS Spectrum obtained by MS/MS readout of the digital polymer. Note that in ideal conditions, the peaks are supposed to correspond to the masses of string fragments, or more precisely, masses of prefixes and suffixes of the string. Due to measurement errors, spurious peaks arise and one needs to apply specialized signal processing techniques to identify the correct peaks.}
\label{fig:spectrometry}
\vspace{-0.18in}
\end{figure}
Ideally, the masses of all prefixes and suffixes are recovered reliably, allowing one to read the message content by taking the differences of the increasing fragment masses and mapping them to the masses of the $0$ or $1$ symbol. Polymer synthesis is cost- and time-efficient and MS/MS sequencers are significantly faster than those designed for other macromolecules, such as DNA.
Nevertheless, despite the fact that the masses of the polymers can be tuned to allow for more accurate mass discrimination, polymer-based storage systems still suffer from large read error-rates. This is due to the fact that MS/MS sequencing methods tend to produce peaks, representing the masses of the fragments that are buried in analogue noise due to atom disassociation during the fragmentation process.

In an earlier line of work, the authors of~\cite{acharya2014string} introduced the problem of \emph{binary string reconstruction from its substring composition multiset} to address the issue of MS/MS readout analysis. The substring composition multiset of a binary string is obtained by writing out all substrings of the string of all possible length and then representing each substring by its composition. As an example, the string $101$ contains three substrings of length one - $1$, $0$, and $1$, two substrings of length 2 - $10$ and $01$, and one substring of length three - $101$. The composition multisets of the substrings of length one are $\{ 0,1,1 \}$, of length two are $\{0^11^1,0^11^1\}$ and of length three $\{0^11^2\}$. Note that composition multisets ignore information about the actual order of the bits and may hence be seen as only capturing the information about the ``mass'' or ``weight'' of the string. The problem addressed in~\cite{acharya2014string} was to determine for which string lengths may one guarantee unique reconstruction from an error-free composition multiset up to string reversal. The main results of~\cite[Theorem ~17, ~18, ~20]{acharya2014string} asserts that binary strings of length  $\leq 7$, one less than a prime, or one less than twice a prime are uniquely reconstructable up to reversal.

For our line of work, we also rely on the two modeling assumptions described in~\cite{acharya2014string}:

\textit{Assumption 1.} We can infer the composition of a polymer substring from its mass. As long as the masses chosen for $0$ and $1$ are distinct, and the polymer block length is fixed, this assumption is naturally satisfied. 

\textit{Assumption 2.} When a polymer block is broken down for mass spectrometry analysis, we observe the masses of all its substrings with identical frequency. The masses of all binary substrings of an encoded polymer may be abstracted by the composition multiset of a string, provided that Assumption 1 holds. This assumption deviates from the classical ion series theory in so far that the former only provides information about the masses of the prefixes and suffixes, while the abstraction allows one to observe the masses of all substrings, but without a priori knowledge of their order. 

Unlike the work in~\cite{acharya2014string} which has solely focused on the problem of unique string reconstruction, 
we view the problem from a coding-theoretic perspective and ask the following:

\textbf{Q1.} \emph{Can one add asymptotically negligible redundancy to information strings in such a way that unique reconstruction is possible, independent on the length of the strings?} Since only strings of specific lengths are reconstructable up to reversals, we aim to devise an efficient coding scheme that encode all strings of length $k \geq 1$ into strings of a larger length $n \geq k$ that are uniquely and efficiently reconstructable for \emph{all possible string lengths}. Furthermore, we do not allow for both a string and its reversal to be included in the codebook. One simple (non-constructive) means to ensure that a string is uniquely reconstructable up to reversal is to pad the string with bits up to the shortest length of the form $\min\{{p-1,2q-1\}}$, where $p$ and $q$ primes. For example, if $k > 89693$, it is known that there exists a prime $p$ such that  
$k-1 < p-1 < \left( 1+\frac{1}{\ln^3\, k}\right)\,k-1.$
Unfortunately, the result only holds for very large $k$ that are beyond the reach of polymer chemistry. Bertrand's postulate~\cite{hardy1929introduction}, applies for shorter lengths $k>3$, but only guarantees that $k-1 < p-1 < 2k-4.$ This implies a possible code rate reduction to $1/2$. Also, eliminating reversals of strings reduces the codebook by less than a half. 

\textbf{Q2.} \emph{Can one add asymptotically negligible redundancy to information strings in such a way that unique reconstruction is possible even in the 
presence of errors, independent on the length of the strings?} For simplicity, we focus on the single deletion-insertion error model, under which the composition (mass) of one substring is erroneously interpreted as a different composition (mass). 

We answer both questions affirmatively by describing a coding scheme that allows for unique reconstruction and correction of a single deletion-insertion mass error. Encoding is performed by interleaving symmetric strings with Catalan-type paths, while decoding is accomplished through a modification of the backtracking decoding algorithm described in~\cite{acharya2014string}. Our work extends the existing literature in coded string reconstruction~\cite{kiah2016codes,gabrys2018unique}.

\vspace{-0.1in}
\section{Problem Statement} \label{ps}
\vspace{-0.05in}
Let $\textbf{s}=s_1 s_2 \ldots s_k$ be binary a string of length $k \geq 2$. A substring of $\textbf{s}$ starting at $i$ and ending at $j$, where $1 \leq i < j \leq k,$ is denoted by $\textbf{s}_{i}^{j}$, and is said to have \emph{composition} $0^{z}1^{w}$, where $0 \leq z,w \leq j-i+1$ stand for the number of $0$s and $1$s in the substring, respectively. Note that the composition only conveys information about the weight of the substring, but not the particular order of the bits. Furthermore, let $C_{\ell}(\textbf{s})$  stand for the multiset of compositions of substrings of $\textbf{s}$ of length $\ell$, $1\leq \ell \leq k$; clearly, this multiset contains $k-\ell+1$ compositions. For example, if $\textbf{s}=100101$, then the substrings of length two are $10,00,01,10,01$, so that $C_2(\textbf{s})=\{{0^11^1,0^2,0^11^1,0^11^1,0^11^1\}}$. =

The multiset $C(\textbf{s})=\cup_{\ell=1}^{k} C_{\ell}(\textbf{s})$ is termed the \emph{composition multiset}. Clearly, the composition multisets of a string $\textbf{s}$ and its reversal, $\textbf{s}^r=s_k s_{k-1} \ldots s_1$ are identical and hence these two strings are indistinguishable based on $C(\cdot)$.
We define the \emph{cummulative weight} of a composition multiset $C_\ell(\textbf{s}),$ with compositions of the form $0^{z}1^{w}$, where $z+w=\ell$, as $w_{\ell}(\textbf{s})=\sum_{0^{z}1^{w} \in C_{\ell}(\textbf{s})}\, w.$ Observe that $w_{1}(\textbf{s})=w_{k}(\textbf{s})$, as both equal the weight of the string $\textbf{s}$. 
More generally, one also has $w_{\ell}(\textbf{s})=w_{k-\ell+1}(\textbf{s}), \text{ for all } 1 \leq \ell \leq k.$
In our subsequent derivations, we also make use of the following notation. For a string $\textbf{s}=s_1 s_2 \ldots s_k$, we let $\sigma_i=\text{wt}(s_is_{k-i+1})$ for $i \leq \lfloor \frac{n}{2} \rfloor,$ and $\sigma_{ \lceil \frac{n}{2} \rceil}=\text{wt}(s_{ \lceil \frac{n}{2} \rceil})$, where $\text{wt}$ stands for the weight of the string. For our running example $\textbf{s}=100101,$ $\sigma_1=2,$ while $\sigma_2=0$. We use $\Sigma^{ \lceil \frac{n}{2} \rceil}$ to denote the set $\{ \sigma_i\}_{i \in [ \lceil \frac{n}{2} \rceil]},$ where $[a]=\{{1,\ldots,a\}}$.

Whenever clear from the context, \emph{we omit the argument $\textbf{s}$ and the floors/ceiling functions required to obtain appropriate integer lengths.}

The two problem of interests are as follows. The first problem pertains to reconstruction codes: a collection of binary strings of fixed length is called a \textbf{reconstruction code} if all the strings in the code can be reconstructed uniquely based on their multiset compositions. We seek reconstruction codes of small redundancy and consequently, large rate.

In the second problem, one is given a valid composition multiset of a string $\textbf{s}$, $C(\textbf{s})$. Within the multiset $C(\textbf{s})$, only one composition is arbitrarily corrupted. We refer to such an error as a \textbf{single composition error}, or single insertion-deletion pair. For example, when $\textbf{s}=100101$, the multiset 
$C_2(\textbf{s})=\{{0^11^1,0^2,0^11^1,0^11^1,0^11^1\}}$ may be corrupted to $C_2(\textbf{s})=\{{\mathbf{0^2},0^2,0^11^1,0^11^1,0^11^1\}}$. 
Single composition errors for strings of even length are detectable, since if an error occurs in only one of the two sets $C_{\ell}$ or $C_{k+1-\ell}$, then $w_{\ell} \neq w_{k+1-\ell}$. We seek reconstruction codes capable of correcting one composition error. 

Our main results are summarized below. 
\begin{theorem}
There exist efficiently encodable and decodable reconstruction codes with information string length $k$ and redundancy at most $ \frac{1}{2}\,\log \,(k) +6$. 
\end{theorem} 

\begin{theorem}
There exist efficiently encodable and decodable reconstruction code with information string length $k$ capable of correcting a single composition error and redundancy at most $\frac{1}{2} \log\,(k) +9$. 
\end{theorem} 
\vspace{-0.12in}
\section{Some Technical Background} \label{pb}
 \vspace{-0.04in}
Our codebook design relies on the backtracking algorithm~\cite{acharya2014string}, motivated by the Turnpike problem.
We provide an example illustrating the operation of the algorithm.

\begin{example} Let $\textbf{s} = 1010001010$. It can be shown that the set $\Sigma^5 = \{\sigma_1 = 1,\sigma_2 = 1,\sigma_3=1,\sigma_4=1,\sigma_5=0 \}$ is uniquely determined from the composition multiset. For example, $\sigma_1=1$ can be deduced from the two compositions of length $9$, $0^51^4$ and $0^61^3$. How to determine $\Sigma^{k/2}$ from the composition multiset will be discussed in more detail in the next section. Backtracking starts by determining the first and last bit of the string and then proceeding with inward bit placements. In our example, $s_1=1$ and $s_{10}=0$. From $\Sigma^5$, we easily see that one composition of length $8$ equals $0^51^3$; removing this set from $C_8$ allows us to determine $\{ \emph{wt}{(\textbf{s}_1^8}), \emph{wt}(\textbf{s}_{3}^{10}) \}$. Given $C$ and the previous information, we deduce that $s_2=0$ and $s_9=1$. Note that these values were determined correctly since $\emph{wt}(s_{1}) \not = \emph{wt}(s_{10})$. The same steps can be repeated iteratively, but in general, the algorithm will only be able to determine the compositions of the prefix/suffix extensions, but not their actual placement. This phenomenon can be observed in the next step, since the weights of the currently available prefix and suffix are equal. In this case, the algorithm makes an arbitrary assignment. For instance, the algorithm could make the assignments $\textbf{s}_1^3 = 100$ and $\textbf{s}_{8}^{10}=110$. Nevertheless, at some point, combining the information in $\Sigma^5$ with the current estimate of the prefix and suffix may produce an invalid composition. In this case, the algorithm backtracks to the first position at which an arbitrary assignment was made and reverses it. Thus, the algorithm will backtrack depending on the weights of the prefixes and suffixes of the same length.
\end{example}

\begin{theorem*} \emph{\cite[Theorem ~32]{acharya2014string}} 
Let $ \ell_s \overset{\text{def}}{=} | \{ i \leq n/2: \emph{wt}(\textbf{s}_1^i) = \emph{wt}(\textbf{s}_{n+1-i}^{n}) \text{ and } s_{i+1} \not  = s_{n-i} \} |, $
$ E_s \overset{\text{def}}{=} \{ \textbf{t}: C(\textbf{t})=C(\textbf{s})\},\; \ell_s^{*} \overset{\text{def}}{=} \max_{t \in E_s} \ell_{t}.
$ 
For a given input $C(\textbf{s})$ and $\ell_s$, the backtracking algorithm outputs a set of strings that contains $\textbf{s}$ in time $\mathcal{O}(2^{\ell_s}n^{2} \log\,(n))$. Furthermore, $E_s$ can be recovered in time $\mathcal{O}(2^{\ell^*_s}n^{2} \log\,(n))$.
\end{theorem*}
Clearly, if the string has a length that does not allow for unique reconstruction, the algorithm will return a set of strings and in the process backtrack multiple times. Backtracking is possible even when the string is uniquely reconstructable, and one condition that ensures non-backtracking is to impose the constraint that no prefix has a matching suffix of the same length and same weight. To see how such strings may be constructed, we introduce strings related to Catalan paths.

\begin{theorem}\text{\emph{(Bertrand [1887])}}
Among all strings comprising $a$ $0$s and $b$ $1$s, where $a \geq b$, there are ${a+b \choose a} - {a+b \choose a+1}$ strings in which every prefix has at least as many $0$s as $1$s. Note that when $a=b=h$, ${a+b \choose a} - {a+b \choose a+1} = \frac{1}{h+1} {2h \choose h} =C_h$. The number $C_h$ is known as the $h^{\text{th}}$ \emph{Catalan number}. The central binomial coefficient ${2h \choose h}$, among other things, also counts the number of strings of length $2h$ whose every prefix contains more $0$s than $1$s. We refer to such strings as \emph{Catalan-type}.
\end{theorem} 

The following bounds on the central binomial coefficient will be useful in our subsequent derivations. 
\begin{proposition} \label{prop2}
\emph{The central binomial coefficient may be bounded as:}
\begin{equation}
    \frac{2^{2h}}{\sqrt{\pi h}} \left( 1- \frac{1}{8h}\right) \leq {2h \choose h} \leq \frac{2^{2h}}{\sqrt{\pi h}} \left( 1- \frac{1}{9h}\right),\; \forall \, h \geq 1.
\end{equation}
\end{proposition}

\section{Reconstruction Codes} \label{sec:recons}

In what follows, we describe a family of efficiently encodable and decodable reconstruction codes that map strings of any length $k$ into strings of length $n \leq k+  1/2 \log\,(k) +6$. 

Using $C_1$ and recalling that $\sigma_i = \text{wt}(s_i, s_{n+1-i})$, we have $\sum^{n/2}_{j=1} \sigma_j = w_1.$
When $i=2$, the bits at positions $1,n$ contribute once to $w_2$, whereas the bits $2, \dots, n-1$ all contribute twice to $w_2$. Using $C_2,$ we hence get $\sigma_1 + 2 \sum^{ n/2}_{j=2} \sigma_j = w_2.$ 
Generalizing for all $C_i, i\leq n/2$, we have 
\begin{equation} \label{eq:sigmas}
\frac{1}{i} \sigma_1 + \frac{2}{i} \sigma_2 + \dots + \frac{i-1}{i} \sigma_{i-1} +  \sigma_i + \sigma_{i+1} + \dots +  \sigma_{n/2} = \frac{1}{i} w_{i}. 
\end{equation}
This gives a system of $n/2$ linear equations with $n/2$ unknowns that can be solved efficiently. Thus, for all error-free composition sets, one can find $\Sigma^{n/2}$. Therefore, the problem of interest is to determine $\textbf{s}$ provided $\Sigma^{n/2}$ and $C(\textbf{s})$. 
~\cite[Lemma~31]{acharya2014string} asserts that when $\text{wt}(\textbf{s}_1^i) \not = \text{wt}(\textbf{s}^{n}_{n +1 -i})$, then $C(\textbf{s}), \textbf{s}_1^i,$ and $\textbf{s}_{n-i+1}^n$ determine the ordered pair $(s_{i+1},s_{n-i})$. 

The previous lemma will be used to guide our construction of reconstructible code based on Catalan-type strings. We proceed as follows. Let $I \subseteq [n]$. The string formed by concatenating bits at positions in $I$ in-order is denoted by $\textbf{s}|_{I}$. To construct a string $\textbf{s}$ of a reconstruction code $\mathcal{S}_R(n)$ of even length $n$ we proceed as follows.

\begin{align}
    \mathcal{S}_R(n) = &\{ \textbf{s} \in \{ 0,1\}^{n}, s_1 =0, s_n =1, \label{set1} \\
 & \exists \; I \subseteq \{ 2, \dots, n-1\} \text{ such that} \nonumber \\ 
 & \qquad \qquad \qquad \quad \quad \quad \quad \text{ for all } i \in I, s_i \not = s_{n+1-i}, \nonumber\\
 & \qquad \qquad \qquad \quad \quad \quad \quad  \text{           } \text{for all }  i \not \in I, s_i = s_{n+1-i}, \nonumber \\
                        &\textbf{s}_{[n/2] \cap I}  \text{ is a Catalan-type string.}         \nonumber                \} \end{align}

For $n$ odd, we define the codebook as $\mathcal{S}_R(n)= \{ \textbf{s}_1^{n/2}\; 0 \; \textbf{s}_{n/2+1}^{n}, \; \textbf{s}_1^{n/2} \; 1 \, \textbf{s}_{n/2+1}^{n}, \textbf{s} \in  \mathcal{S}_R(n-1) \} $. 

The following proposition is an immediate consequence of the construction described above.
\begin{lemma} \label{lem1}
Consider a string $\mathbf{s}\in \mathcal{S}_R(n)$. For all prefix-suffix pairs of length $ 1\leq  j \leq n/2$, one has $\emph{wt}(\textbf{s}_1^j) \not = \emph{wt}(\textbf{s}^{n}_{n +1 -j})$.
\end{lemma}

The encoding algorithm that accompanies our reconstruction codebook can be easily implemented using efficient rankings of Catalan strings and symmetric strings that are ordered lexicographically.

The proof of Theorem~1 follows from the fact that $\mathcal{S}_R(n)$ is a reconstruction code, which may be easily established from the guarantees for the backtracking algorithm and~Lemma \ref{lem1}.

The size of $\mathcal{S}_R(n)$ may be simply bounded as: 
\begin{align*}
    |\mathcal{S}_R(n)| \geq \frac{1}{2}\sum_{i=0}^{(n-2)/2} {\frac{n-2}{2} \choose i} 2^{\frac{n-2}{2} -i} {i \choose  \frac{i}{2}} 
     \geq \frac{3 \, 2^{n-5} }{\sqrt{2 \pi (n-2)}} \, . 
\end{align*}
The first inequality follows from the description of the codebook, while the second follows from Proposition \ref{prop2} and the binomial theorem. As $2^k \leq |\mathcal{S}_R(n)|$, simple algebraic manipulation reveals that the redundancy of the reconstruction code for information lengths $k$ is at most 
$1/2  \log\,(k) +6$.   

\section{Error-Correcting Reconstruction Codes} \label{sec:step1}

Our single composition error-correcting codes use the same interleaving procedure described in the previous section, but require adding a constant number of redundant bits. 
In particular, let $\mathcal{S}_R(n-2)$ be the code of odd length $n-2$ described in the previous section. Then, a single composition error-correcting code $\mathcal{S}_C(n)$ is constructed by adding two bits to each string in $\mathcal{S}_R(n-2)$ and subsequently fixing the value of one additional bit. These three redundant bits allow us to uniquely recover the set $\Sigma^{ n/2 }$ in the presence of a single composition error. Consequently, Lemma~\ref{lem:3} can be used to show that given $\Sigma^{ n/2 }$ and the erroneous composition set of $\textbf{s}$, one can reconstruct $\textbf{s}$. 

To prove Theorem~2, let $C'$ denote the set obtained by introducing a single error in the composition set $C(\textbf{s})$ of a string $\textbf{s}$. Furthermore, let $w'_j$ denote the cumulative weight of compositions in $C'_j$, and recall that $w_j$ stands for the cumulative weight of compositions in $C$, such that $w_j = w_{n-j+1}$. It is straightforward to prove the following proposition.

\begin{proposition} Let $j \in [n]$. Then, 
$$j w_1 -   \sum_{i=1}^{j-1} i \, \sigma_{j-i} - 2 \leq  w_j \leq j w_1 -   \sum_{i=1}^{j-1} i \, \sigma_{j-i} .$$
\end{proposition}
This result immediately implies the next proposition.
\begin{proposition}\label{prop:cor} Let $j \in [n]$ and suppose that we are given $w_1, \sigma_1, \ldots, \sigma_{j-1}$. Then, the value $w_j \bmod 3$ uniquely determines $w_j$.
\end{proposition}

We also need the following three propositions. 

\begin{proposition}\label{prop:w1} Given $\emph{wt}(\textbf{s}) \bmod 2,$ $w_n'$ and $w_1'$, one can recover $w_1$. 
\end{proposition}
\begin{proof} If $w_n' = w_1'$, then clearly $w_1=w_n' = w_1'$. Hence, suppose that $w_n' \neq w_1'$ and observe that $|w_1' - w_1| \leq 1$. 
The last inequality follows since at most one composition error is allowed. If $w'_1 \bmod 2 = \text{wt}(\textbf{s}) \bmod 2$, then $w_1 =  w'_1$; 
otherwise, $w_1 = w_n'$.
\end{proof}

\begin{proposition}\label{prop:change} Suppose that $n$ is odd and that either $ \lceil \frac{n}{2} \rceil + 1$ or $ \lceil \frac{n}{2} \rceil$ is divisible by $3$. Assume that $\textbf{s} = s_1 \, \ldots \, s_{ \lceil \frac{n}{2} \rceil} \, \ldots s_n,$ and let $\textbf{s}' = s_1 \, \ldots \, 1-s_{\lceil \frac{n}{2} \rceil} \, \ldots \, s_n$. Then,
\begin{align*}
\sum_{i=1}^{ \lceil \frac{n}{2} \rceil} w_i(\textbf{s}) \equiv \sum_{i=1}^{\lceil \frac{n}{2} \rceil} w_i(\textbf{s}') \bmod 3.
\end{align*}
\end{proposition}
\begin{proof} Suppose that $s_{\lceil \frac{n}{2} \rceil}=1$. Then, the bit $s_{ \lceil \frac{n}{2} \rceil}$ contributes 
$ \lceil \frac{n}{2} \rceil$ to $w_{ \lceil \frac{n}{2} \rceil}$ and $ \lceil \frac{n}{2}  \rceil-1$ to $w_{ \lceil \frac{n}{2}\rceil -1}$. In summary, 
if $s_{\lceil \frac{n}{2} \rceil}=1$, then 
$$ \sum_{i=1}^{ \lceil \frac{n}{2} \rceil} w_i(\textbf{s}) = \sum_{i=1}^{ \lceil \frac{n}{2} \rceil} w_i(\textbf{s}') + \frac{ \lceil \frac{n}{2} \rceil  \, ( \lceil \frac{n}{2} \rceil +1)}{2}.$$
The result follows if either $ \lceil \frac{n}{2} \rceil + 1$ or $ \lceil \frac{n}{2} \rceil $ is divisible by $3$.
\end{proof}

\begin{proposition}\label{prop:change2} For odd $n,$ if $s_1 \, \ldots \, s_{ \lceil \frac{n}{2} \rceil} \, \ldots \, s_n \in \mathcal{S}_R(n)$, 
then $s_1 \, \ldots \, 1-s_{ \lceil \frac{n}{2} \rceil} \, \ldots \, s_n \in \mathcal{S}_R(n)$.
\end{proposition}

Our code for odd $n$ is defined as follows (an almost identical construction is valid for even $n$):
\begin{align*}
\mathcal{S}_C(n) = \Big \{ &\textbf{s}=s_1\, s^*_1\, s_2 \ldots \, s_{ \lceil \frac{n-2}{2} \rceil} \, \ldots \, s_{n-3} \, s^*_n \, s_{n-2} \in \{0,1\}^n :  \\
&\ \ \ \ \ s_1 \, \ldots \, s_{n-2} \in \mathcal{S}_R(n-2), \text{wt}(\textbf{s}) \bmod 2 \equiv 0,   \\
&\ \ \ \ \ \sum_{i=1}^{ \frac{n}{2} } w_i(\textbf{s}) \equiv 0 \bmod 3, \, \text{ where } s^*_1 \leq s^*_n \Big \}. 
\end{align*}

The size of the code $\mathcal{S}_C(n)$ is $\frac{|\mathcal{S}_R(n-2)|}{2}$, which follows since we removed one information symbol from each coded string in $S_R(n-2)$ by requiring $\text{wt}(\textbf{s}) \bmod 2 \equiv 0,$ and then added two more redundant symbols. 
To construct a string in $\mathcal{S}_C(n)$, we first fix $s^*_1$ and $s^*_n$ so that $\sum_{i=1}^{ \lceil \frac{n}{2} \rceil} w_i(\textbf{s}) \equiv 0 \bmod 3$. Then, we choose $s_{ \lceil \frac{n-2}{2} \rceil}$ to satisfy $\text{wt}(\textbf{s})\equiv 0 \bmod 2 $. From Propositions~\ref{prop:change} and \ref{prop:change2}, the resulting string belongs to $\mathcal{S}_C(n)$. 

For the next lemma, recall that $C'(\textbf{s})$ is the result of a single composition error in $C(\textbf{s})$.

\begin{lemma} Suppose that $\textbf{s} \in \mathcal{S}_C(n)$. Then, given $C'(\textbf{s})$, one can recover $\Sigma^{ n/2 }$. \label{finding_sigma}
\end{lemma}
\begin{proof} In order to prove the claim, we show that given $C'(\textbf{s}),$ one can recover $w_1, w_2, \ldots, w_n,$ which we know uniquely determine $\Sigma^{ n/2 }$ according to (\ref{eq:sigmas}). Let $j$ be such that $w'_j \neq w_{n+1-j}'$. Since at most one 
single composition error is allowed, there exists at most one such $j$. It is straightforward to see that due to symmetry, 
either $w'_j \neq w_j=w_{n+1-j}$ or $w'_{n+1-j} \neq w_j=w_{n+1-j}$. Since $\text{wt}(\textbf{s}) \bmod 2 \equiv 0$ by construction, 
it follows that we can determine $w_1$ based on Proposition~\ref{prop:w1}. Then, according to Proposition~\ref{prop:cor}, we can recover 
$w_j$ and all of $w_1, \ldots, w_n$. One case left to consider is when $w'_i = w_{n+1-i}'$ for all $i$. In this case, $w'_{ \frac{n}{2} } \neq w_{ \frac{n}{2} }$. Applying Proposition~\ref{prop:cor} allows us to determine $w_{ \frac{n}{2} }$ for this case as well, and this completes the proof.
\end{proof}

Next, let $\mathcal{T}_i$ be the set of compositions of all substrings $\textbf{s}_j^k$ for which $j<k\leq i,$ or $n+1-i \leq j<k,$ or $j \leq i \textit{ and } n+1-i \leq k$. 

\begin{lemma} \label{lem:3}
Let $\textbf{s} \in \mathcal{S}_C(n)$.
Given $C'(\textbf{s})$, one can uniquely reconstruct the string $\textbf{s}$. 
\end{lemma}

\begin{proof}
Let $j$ denote the index of the composition multi-set $C_j$ that contains an error. From Lemma~\ref{finding_sigma}, $\Sigma^{n/2}$ may be determined in an error-free manner. Using the obtained $\Sigma^{n/2}$, we run the backtracking algorithm and in the process, we may run into non-compatible compositions for $j> \frac{n}{2} $. For the case that backtracking halts for $j = n-i-1$, the currently reconstructed sub-strings are $\textbf{s}^i_1, \textbf{s}^n_{n+1-i}$. Without loss of generality, assume that $\sigma_{i+1}=1$ as otherwise one can fix the error easily. Furthermore, note that $\mathcal{T}_i$ can be constructed from $\Sigma^{n/2},\textbf{s}_1^i, \text{ and } \textbf{s}^{n}_{n +1 -i} $. 

One way in which incompatibility may manifest itself is through $\mathcal{T}_i \not \subset C'$, where $j = n-i-1$. In this case, we identify the element that is in $\mathcal{T}_i$ but not in $C'_j$, and add its weight to $w'_j$ and compare it with $w'_{n+1-j}$; this allows us to identify the erroneous composition. Next, suppose that $\mathcal{T}_i \subset C'$. In this case, consider the two longest compositions in $C' \setminus \mathcal{T}_i$. The two longest compositions in $C' \setminus \mathcal{T}_i$ are the compositions of a prefix-suffix pair of length $j$. 
Since we have reconstructed the prefix and suffix of length $i$ and we know that $\sigma_{i+1} = 1$, there are two possibilities for compositions compatible with the prefix and two for the suffix of length $i+1$. Out of the six pairs of compositions that may be chosen from the four compositions, only two pairs cannot be directly eliminated as candidates for the correct composition. In this case, the following two prefix-suffix substrings are possible: $\{\textbf{s}_1^i \, 0 , 1\, \textbf{s}^{n}_{n-i+1} \}, \{\textbf{s}_1^i \, 1 , 0 \, \textbf{s}^{n}_{n-i+1} \}$. To show that only one of the constructed prefix-suffix pairs will be valid (compatible), it suffices to show the following: For any two strings $\textbf{s}_1, \textbf{s}_2 \in \mathcal{S}_C(n)$ that have the same $\Sigma^{n/2}$, $|C(\textbf{s}_1) \setminus C(\textbf{s}_2)| \geq 4$.

Let us assume that on the contrary, there are two strings $\textbf{s}, \textbf{t}$ such that $|C(\textbf{s}) \setminus C(\textbf{t})|=2$, and that they differ only in their respective $C_j$ sets (this condition is imposed by the Catalan strings, see Figure~\ref{fig:confusable_strings}). 

Since the prefixes and suffixes of the strings of length $i=n-j-1$ are identical, we let $s_1, \dots, s_i$ and $s_{n+1-i}, \dots, s_n$ denote the first and last $i$ bits of both strings. Let $c(\textbf{s})$ denote the composition of the string $\textbf{s}$. Furthermore, let $c(\textbf{s}_{l}^{l'})$ denote the composition of $\textbf{s}_{l}^{l'}$, $l\leq l'$.  

When $n=2(i+1)+1$, the strings differ in two compositions in $C_{n+1-i}$ due to the above observations. Note that they also differ in two compositions in their respective multisets $C_{i}$. 

When $n \geq 2(i+1) + 3$ and $\sigma_{i+2}=1$, we let $b_{s}$ stand for the $(i+2)^{\text{th}}$ bit in the string $\textbf{s}$, and $b_t$ stand for the $(i+2)^{\text{th}}$ bit of string $\textbf{t}$. When $\sigma_{i+2} \in  \{0, 2 \}$, we let $b$ denote the $(i+2)^{\text{th}}$ bits of the two strings, which are identical.
Next, we determine conditions under which $C_{j-1}(\textbf{s})= C_{j-1}(\textbf{t})$. Note that the compositions of substrings of length $n-i-2$ that contain the bits $i+1,\dots,n-i$ are identical for the two strings. 
\begin{figure}[h]
\centering
\includegraphics[scale=0.29]{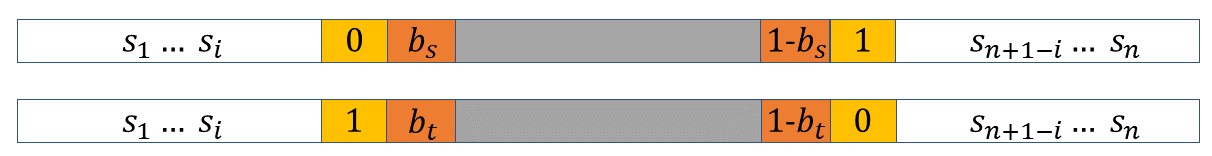} 
\caption{The figure depicts two strings $\textbf{s}, \textbf{t}$ satisfying the assumptions used in the proof.}
\label{fig:confusable_strings}
\vspace{-0.08in}
\end{figure}\\
\emph{Case 1}: $\sigma_{i+2} = 1$. With a slight abuse of notation, we choose to write compositions as sets containing both bits and other compositions. On the left-hand-side of the equation below, the compositions correspond to the substrings of $\textbf{s}$ of length $n-i-2$ that \emph{may} differ for the two strings. 
The right-hand-side of the equation corresponds to the same entities in $\textbf{t}$. If the equation holds, then the multisets $C_{j-1}(\textbf{s})$ and $C_{j-1}(\textbf{s})$ are equal.
\begin{equation*}
\begin{rcases}
    \begin{dcases}
\{{c(\textbf{s}_1^i), 0, b_s, c\}}, \\
\{{c(\textbf{s}_2^i), 0, b_s, c ,1-b_s\}}, \\
\{{c(\textbf{s}_{j+2}^n),1, 1-b_s, c\}},  \\
\{{c(\textbf{s}_{j+2}^{n-1}), 1, 1-b_s, c, b_s\}}
\end{dcases} 
\end{rcases}
=
\begin{rcases}
    \begin{dcases}
\{{c(\textbf{s}_1^i),1, b_t, c\}}, \\
\{{c(\textbf{s}_2^i), 1, b_t, c, 1-b_t\}}, \\
\{{c(\textbf{s}_{j+2}^n), 0, 1-b_t, c\}},  \\
\{{c(\textbf{s}_{j+2}^{n-1}), 0, 1-b_t, c, b_t\}}
	\end{dcases} 
\end{rcases} 
\end{equation*}
Due to space limitations, we omit the exhaustive case-by-case arguments that show that the above 
set equality is never true, independently on how $b_s$ and $b_t$ are chosen.

\emph{Case 2}: $\sigma_{i+2} \in \{0,2 \}$  Similar reasoning leads to a set equality condition 
in which $b_s$ and $b_t$ are replaced by $b$. Once again, it can be shown by an exhaustive 
case-by-case analysis that the set equality never holds, independently on the choice of $b$. 
This implies that the composition sets $C_{j-1}(\textbf{s})$ and $C_{j-1}(\textbf{t})$ differ, 
which in turn implies that the composition multisets of the two strings are at distance $\geq 4$. 
\end{proof}
The backtracking string reconstruction process based on an erroneous composition set is straightforward: It takes $\mathcal{O}(n^2)$ time to compute the $\mathcal{T}_k$ multiset, and backtracking performs $\mathcal{O}(n)$ steps. Thus, the decoding algorithm can computes the original string in $\mathcal{O}(n^3)$ time. 
\vspace{-0.1in}
\bibliography{biblio} 

\begin{thebibliography}{10}

\bibitem{zhirnov2016nucleic}
V.~Zhirnov, R.~M. Zadegan, G.~S. Sandhu, G.~M. Church, and W.~L. Hughes,
  ``Nucleic acid memory,'' {\em Nature materials}, vol.~15, no.~4, p.~366,
  2016.

\bibitem{al2017mass}
A.~Al~Ouahabi, J.-A. Amalian, L.~Charles, and J.-F. Lutz, ``Mass spectrometry
  sequencing of long digital polymers facilitated by programmed inter-byte
  fragmentation,'' {\em Nature communications}, vol.~8, no.~1, p.~967, 2017.

\bibitem{goldman2013towards}
N.~Goldman, P.~Bertone, S.~Chen, C.~Dessimoz, E.~M. LeProust, B.~Sipos, and
  E.~Birney, ``Towards practical, high-capacity, low-maintenance information
  storage in synthesized dna,'' {\em Nature}, vol.~494, no.~7435, p.~77, 2013.

\bibitem{grass2015robust}
R.~N. Grass, R.~Heckel, M.~Puddu, D.~Paunescu, and W.~J. Stark, ``Robust
  chemical preservation of digital information on {DNA} in silica with
  error-correcting codes,'' {\em Angewandte Chemie International Edition},
  vol.~54, no.~8, pp.~2552--2555, 2015.

\bibitem{yazdi2015rewritable}
S.~H.~T. Yazdi, Y.~Yuan, J.~Ma, H.~Zhao, and O.~Milenkovic, ``A rewritable,
  random-access {DNA}-based storage system,'' {\em Scientific reports}, vol.~5,
  p.~14138, 2015.

\bibitem{yazdi2017portable}
S.~H.~T. Yazdi, R.~Gabrys, and O.~Milenkovic, ``Portable and error-free
  {DNA}-based data storage,'' {\em Scientific reports}, vol.~7, no.~1, p.~5011,
  2017.

\bibitem{acharya2014string}
J.~Acharya, H.~Das, O.~Milenkovic, A.~Orlitsky, and S.~Pan, ``String
  reconstruction from substring compositions,'' {\em arXiv preprint
  arXiv:1403.2439}, 2014.

\bibitem{hardy1929introduction}
G.~H. Hardy, ``An introduction to the theory of numbers,'' {\em Bull. Amer.
  Math. Soc.}, vol.~35, pp.~778--818, 11 1929.

\bibitem{kiah2016codes}
H.~M. Kiah, G.~J. Puleo, and O.~Milenkovic, ``Codes for {DNA} sequence
  profiles,'' {\em IEEE Transactions on Information Theory}, vol.~62, no.~6,
  pp.~3125--3146, 2016.

\bibitem{gabrys2018unique}
R.~Gabrys and O.~Milenkovic, ``Unique reconstruction of coded sequences from
  multiset substring spectra,'' in {\em 2018 IEEE International Symposium on
  Information Theory (ISIT)}, pp.~2540--2544, IEEE, 2018.

\end{thebibliography}
\bibliographystyle{ieeetr}

\end{document}